\documentclass[12pt]{article}
\usepackage{amsmath, amsthm, amssymb, setspace, mathrsfs, color, verbatim, graphicx,hyperref, amsfonts, bm, rotating, lscape, multicol, multirow,bbm,natbib}
\usepackage[left=0.7in,right=1.0in,top=0.7in,bottom=1.0in,nohead]{geometry}
\pdfoutput=1
\newtheorem{thm}{Theorem}
\newtheorem{asm}[thm]{Assumption}
\newtheorem{lem}[thm]{Lemma}

\def\bse{\begin{eqnarray*}}
\def\ese{\end{eqnarray*}}
\def\be{\begin{eqnarray}}
\def\ee{\end{eqnarray}}

\newcommand{\utwi}[1]{\mbox{\boldmath $ #1$}}
\newcommand{\Y}{{\utwi{Y}}}
\newcommand{\X}{{\utwi{X}}}
\providecommand{\e}[1]{\ensuremath{\times 10^{#1}}}
\DeclareMathOperator*{\argmax}{argmax}
\DeclareMathOperator*{\argmin}{argmin}

\begin{document} 

\title{A Nonparametric Approach for Multiple Change Point Analysis of Multivariate Data}
\author{David S. Matteson and Nicholas A. James \\ Cornell University\footnote{
Matteson is an Assistant Professor, 
Department of Statistical Science,
Cornell University,
1196 Comstock Hall,
Ithaca, NY 14853
(Email: \href{mailto:matteson@cornell.edu}{matteson@cornell.edu}; Web: \url{http://www.stat.cornell.edu/\~matteson/}). 
\hspace{1cm}
James is a PhD Candidate, 
School of Operations Research and Information Engineering,
Cornell University,
206 Rhodes Hall,
Ithaca, NY 14853
(Email: \href{mailto:nj89@cornell.edu}{nj89@cornell.edu}; Web: \url{https://courses.cit.cornell.edu/nj89/}).}}
\date{\today}
\maketitle


\begin{abstract}

Change point analysis has applications in a wide variety of fields. The general problem concerns the inference of a change in  distribution for a 
set of time-ordered observations. Sequential detection is an online version in which new data is continually arriving and is analyzed adaptively. We 
are concerned with the related, but distinct, offline version, in which retrospective analysis of an entire sequence is performed. For a set of 
multivariate observations of arbitrary dimension, we consider nonparametric estimation of both the number of change points and the positions at which they occur. We do not 
make any assumptions regarding the nature of the change in distribution or any distribution assumptions beyond the existence of the $\alpha$th absolute moment, for some $\alpha \in (0,2)$. 
Estimation is based on hierarchical clustering and we propose both divisive and agglomerative algorithms. 
The divisive method is shown to provide consistent estimates of both the number and location of change points under standard regularity assumptions. 
We compare the proposed approach with competing methods in a simulation study. 
Methods from cluster analysis are applied to assess performance and to allow simple comparisons of location estimates, even when the estimated number differs. 
We conclude with applications in genetics, finance and spatio-temporal analysis.

\end{abstract}

\par\vfill\noindent
{\bf KEY WORDS:}
Cluster analysis;
Multivariate time series;
Permutation tests;
Signal processing;
$U$-statistics.
\par\medskip\noindent
{\bf Short title: Nonparametric Change Point Analysis}

\clearpage\pagebreak\newpage
\newlength{\gnat}
\setlength{\gnat}{24pt} 
\baselineskip=\gnat

\section{Introduction}\label{introduction}
Change point analysis is the process of detecting distributional changes within time-ordered observations. This arises in financial modeling 
\citep{Talih:2005}, where correlated assets are traded and models are based on historical data represented as multivariate time series. It is applied in bioinformatics \citep{Muggeo:2011} 
to identify genes that are associated with specific cancers and other diseases. Change point analysis is also used to detect credit card fraud 
\citep{Bolton:2002} and other anomalies \citep{Sequeira:2002, Akoglu:2010}; and for data classification in data mining \citep{Mampaey:2011}. 
Applications can also be found in signal processing, where change point analysis can be used to detect significant changes within a stream of images \citep{Kim:2009}.

While change point analysis is important in a variety of fields, the methodologies that have been developed to date often assume a single or known number of change points. This assumption is often unrealistic, as seen in Section \ref{real-data}. Increasingly, applications also require detecting changes in multivariate data, for which traditional methods have limited applicability. 
To address these shortcomings, we propose a new methodology, based on $U$-statistics, that is capable of consistently estimating an unknown number of multiple change point locations. The proposed methods are broadly defined for observations from an arbitrary, but fixed dimension.

In general, change point analysis may be performed in either parametric and nonparametric settings. 
Parametric analysis necessarily assumes that the underlying distributions belong to some known family, 
and the likelihood function plays a major role. 
For example, in \cite{Carlin:1992} and \cite{Lavielle:2006} analysis is performed by maximizing a log-likelihood function, while \cite{Page:1954} examines the ratio of log-likelihood functions to estimate change points. 
Additionally, \cite{Davis:2006} combine the log-likelihood, the minimum description length, and a genetic algorithm in order to identify change points. 
Nonparametric alternatives are applicable in a wider range of applications than are parametric ones \citep{Hariz:2007}. Nonparametric approaches often rely heavily on the estimation of density functions \citep{Karahara:2011}, though they have also been performed using rank statistics \citep{Fong:2011}. 
We propose a nonparametric approach based on Euclidean distances between sample observations. It is simple to calculate and avoids the difficulties associated with multivariate density estimation.
 
Change point methods are often directly motivated by specific fields of study. For example, \cite{Johnson:2011} discusses an approach that is rooted in information theory, and ideas from model selection are applied for determining both the number and location of change points in \cite{Yao:1987} and \cite{Zhang:2007}. 
The proposed approach is motivated by methods from cluster analysis \citep{Rizzo:2005}. 

Change point algorithms either estimate all change points concurrently or hierarchically. 
Concurrent methods generally optimize a single objective function.  For example, given that there are $k$ change points, \cite{Hawkins:2001} estimates change point locations by maximizing a likelihood function. \cite{Lavielle:2006} accomplish 
the same task by minimizing a loss function. 
Sequential methods generally estimate change points one at a time \citep{Guralnik:1999}, although some have the ability to estimate two or more at any given stage \citep{Olshen:2004}. 
Such approaches are often characterized as bisection procedures. 
The proposed method utilizes a bisection approach for its computational efficiency. 

We propose a new method that can detect {any} distributional change within an independent sequence, and which does not make any distributional assumptions beyond the existence of the $\alpha$th absolute moment, for some $\alpha \in (0,2)$. 
Estimation is performed in a manner that simultaneously identifies both the number and locations of change points. 
In Section \ref{methodology} we describe our methodology; its properties are discussed in Section \ref{algorithm-analysis}. 
In Sections \ref{simulations} and \ref{real-data} we present the results of our procedure when applied to simulated and real data, respectively. 
In Section \ref{agglomerative} we propose an alternative algorithm and illustrate its use on a novel spatio-temporal application.
Concluding remarks are in Section \ref{conclusion}.

\section{Methodology}\label{methodology}

To highlight the generality of the proposed method, we briefly summarize the different conditions under which analysis may be performed, in increasing complexity. 
Let $Z_1,Z_2,\dots,Z_T\in\mathbb R^d$ be an independent sequence of time-ordered observations.
Throughout this manuscript, the time between observations is assumed positive; it may be fixed or randomly distributed. 
The time index simply denotes the time order. 
In the simplest case, there is a single hypothesized change point location $\tau$.
Specifically, $Z_1,\dots,Z_\tau\stackrel{iid}{\sim} F_1$ and $Z_{\tau+1},\dots,Z_T \stackrel{iid}{\sim} F_2$, in which $F_1$ and $F_2$ are unknown probability distributions. 
Here we test for homogeneity in distribution,  
$H_0: F_1 = F_2$ verses $H_A: F_1 \neq F_2$. For univariate observations with continuous distributions the familiar Kolmogorov-Smirnov test may be applied, and in the general case the approach in \cite{Rizzo:2010} may be applied. 
If $H_0$ is rejected we conclude there is a change point at $\tau$, otherwise we conclude there is no distributional change in the observations. 

A slight modification of the above setting assumes instead that the change point location is unknown, but assumes that at most only one change point exists. 
A natural way to proceed is to choose $\tau$ as the most likely location for a change point, based on some criterion.
Here, $\tau$ is chosen from some subset of $\{1,2,\dots,T-1\}$, then a test for homogeneity is performed. 
This should necessarily incorporate the fact that $\tau$ is unknown. 

Now, suppose there 
is a known number of change points $k$
in the series, but with unknown locations. 
Thus, there exist change points $0 <\tau_1<\cdots<\tau_k<T$, that partition the sequence into 
$k+1$ clusters, such that observations within clusters are identically distributed, and observations between adjacent clusters are not.
 A naive approach for estimating the best of all ${\cal O}(T^k)$ change point locations quickly becomes computationally intractable for $k\ge 3$. 
One remedy is to instead maximize the objective function through the use of dynamic programming as in \cite{Cappe:2007}, 
\cite{Rigaill:2010} and \cite{Fong:2011}.
 
Finally, in the most general case, both the number of change points as well as their locations are unknown. 
Here, the naive approach to concurrent estimation becomes infeasible. As such,
bisection \citep{Vostrikova:1981, Cho:2012} and model selection procedures \citep{Lavielle:2006, Arlot:2012} are popular under 
these conditions.

We now present a nonparametric technique, which we call E-Divisive, for performing multiple change point analysis of a sequence of  multivariate observations. The E-Divisive 
method combines {bisection} 
\citep{Vostrikova:1981} with a multivariate divergence measure from \cite{Rizzo:2005}. 
We first discuss measuring differences in multivariate distributions.
We then propose a procedure for hierarchically estimating change point locations. 
We conclude this section by discussing the hierarchical statistical testing used to determine the number of change points.

\subsection{Measuring Differences in Multivariate Distributions}\label{measuring-differences}
%
For complex-valued functions $\phi(\cdot),$ the complex conjugate of $\phi$ is denoted by $\overline{\phi},$ 
and the absolute square $|\phi|^2$ is defined as $\phi\overline{\phi}.$ 
The Euclidean norm of $x \in \mathbb{R}^d$ is $|x|_d,$ or simply $|x|$ when there is no ambiguity. 
A primed variable such as $X'$ is an independent copy of $X;$ that is, $X$ and $X'$ are independent and identically distributed (iid).

For random variables $X, Y \in \mathbb{R}^{d},$ 
let $\phi_x$ and $\phi_y$ denote the 
characteristic functions of $X$ and $Y$, respectively.
A divergence measure between multivariate distributions may be defined as 
\be\label{dMetric}
\int_{\mathbb{R}^{d}} | \phi_{x}(t) - \phi_{y}(t) |^2 \, w(t) \, dt, 
\ee
in which $w(t)$ denotes an arbitrary positive weight function, for which the above integral exists. 
In consideration of \cite{Rizzo:2005}, we use the following weight function
%
\be\label{w}
w(t; \alpha) = \left( \frac{2 \pi^{d/2}\Gamma(1 - \alpha/2)}{\alpha 2^{\alpha} \Gamma((d+\alpha)/2)} |t|^{d + \alpha} \right)^{-1},
\ee
for some fixed constant $\alpha \in (0,2).$ 
Then, if $E|X|^{\alpha}, E|Y|^{\alpha} < \infty,$ 
a characteristic function based divergence measure 
may be defined as 
\be\label{DXY}
{\cal D}(X, Y; \alpha) = \int_{\mathbb{R}^{d}} | \phi_{x}(t) - \phi_{y}(t) |^2 \,  \left( \frac{2 \pi^{d/2}\Gamma(1 - \alpha/2)}{\alpha 2^{\alpha} \Gamma((d+\alpha)/2)} |t|^{d + \alpha} \right)^{-1} \, dt. 
\ee

Suppose $X, X' \stackrel{iid}{\sim} F_x$ and $Y, Y' \stackrel{iid}{\sim} F_y,$ and that $X, X', Y,$ and $Y'$ are mutually independent. 
If $E|X|^{\alpha}, E|Y|^{\alpha} < \infty,$ then we may employ an alternative divergence measure based on Euclidean distances, defined by \cite{Rizzo:2005} as 
\be\label{EXY}
{\cal E}(X, Y; \alpha) = 2E|X - Y|^{\alpha} - E|X - X'|^{\alpha} - E|Y - Y'|^{\alpha}.
\ee

\begin{lem}\label{thm2}
For any pair of independent random vectors $X, Y \in \mathbb{R}^{d}$, 
and for any $\alpha \in (0,2),$
if $E(|X|^{\alpha} + |Y|^{\alpha}) < \infty,$ 
then ${\cal E}(X, Y; \alpha) = {\cal D}(X, Y; \alpha),$ 
${\cal E}(X, Y; \alpha) \in [0,\infty),$ 
and ${\cal E}(X, Y; \alpha) = 0$ if and only if $X$ and $Y$ are identically distributed.
\end{lem}
For a proof see \cite{Rizzo:2005}, page 178.

The equivalence established in Lemma \ref{thm2} motivates a remarkably simple empirical divergence measure for multivariate distributions based on $U$-statistics.
Let $\X_n = \{ X_i : i = 1,\ldots,n \}$ and $\Y_m = \{ Y_j : j = 1,\ldots,m \}$ be independent iid samples from the distribution of $X, Y \in \mathbb{R}^{d},$ respectively,
such that $E|X|^{\alpha}, E|Y|^{\alpha} < \infty$ for some $\alpha \in (0,2).$ Then an empirical divergence measure analogous to Equation (\ref{EXY}) may be defined as
\be\label{eqn3}
\widehat{{\cal E}}(\X_n, \Y_m; \alpha) = \frac{2}{mn} \sum_{i = 1}^n \sum_{j = 1}^m |X_i - Y_j|^{\alpha} 
- {\binom{n}{2}}^{-1} \hspace{-10pt} \sum_{1 \le i < k \le n} \hspace{-5pt} |X_i - X_k|^{\alpha} 
- {\binom{m}{2}}^{-1} \hspace{-10pt} \sum_{1 \le j < k \le m} \hspace{-5pt} |Y_j - Y_k|^{\alpha}. 
\ee
This measure is based on Euclidean distances between sample elements and is ${\cal O}(m^2 \vee n^2)$, whereas the sample counterpart of Equation (\ref{DXY}) requires $d$-dimensional integration to evaluate. 

Under the assumptions above, $\widehat{\cal E}(\X_n,\Y_m;\alpha)\rightarrow{\cal E}(X,Y;\alpha)$
almost surely as $m \wedge n \to\infty$ by the Strong Law of Large Numbers for $U$-statistics \citep{Hoeffding:1961} and the continuity theorem.
Additionally, under the null hypothesis of equal distributions, i.e.,\ $\mathcal E(X,Y;\alpha) = 0$, we note that 
$\frac{mn}{m+n}\widehat{\cal E}(\X_n,\Y_m;\alpha)$ converges in distribution to a non-degenerate random variable 
as $m \wedge n \to\infty$.
Further, under the alternative hypothesis of unequal distributions, i.e.,\ $\mathcal E(X,Y;\alpha) > 0$, we note that 
$\frac{mn}{m+n}\widehat{\mathcal E}(\X_n,\Y_m;\alpha)\to\infty$ almost surely as $m \wedge n \to\infty$.
These asymptotic results motivate the statistical tests described in Section \ref{significance-test}.

\subsection{Estimating the Location of a Change Point}\label{finding-a-cp}

Let 
\be\label{eqn6}
\widehat{\cal Q}(\X_n,\Y_m;\alpha)=\frac{mn}{m+n}\widehat{\cal E}(\X_n,\Y_m;\alpha)
\ee
denote the scaled sample measure of divergence discussed above. This statistic leads to a consistent approach for estimating change point locations.
Let $Z_1,\dots,Z_T\in\mathbb R^d$ be an independent sequence of observations and let $1\le\tau<\kappa\le T$ be constants. Now define the following sets, $\X_\tau=\{Z_1,Z_2,\dots,Z_\tau\}$ and 
$\Y_\tau(\kappa)=\{Z_{\tau+1},Z_{\tau+2},\dots,Z_{\kappa}\}$. 
A change point location $\hat\tau$ is then estimated as
\begin{equation}\label{find-tauhat}
(\hat\tau,\hat\kappa) =\argmax_{(\tau,\kappa)}\ \widehat{\cal Q}(\X_\tau,\Y_\tau(\kappa);\alpha).
\end{equation}
%
It is possible to calculate the argmax in Equation (\ref{find-tauhat}) in ${\cal O}(T^2)$ by observing that 
$\widehat{\cal Q}(\X_\tau,\Y_\tau(\kappa);\alpha)$ can be derived directly from $\widehat{\cal Q}(\X_{\tau-1},\Y_{\tau-1}(\kappa);\alpha)$ 
and the distances $\{ |Z_{\tau}-Z_j|^\alpha : 1\le j<\tau \}$.

If it is known that at most one change point exists, we fix $\kappa = T$.
Otherwise, the variable $\kappa$ is introduced to alleviate a weakness of {bisection}, as mentioned in \cite{Venkatraman:1992},
in which it may be more difficult to detect certain types of distributional changes in the multiple change point setting using only bisection.
For example, if we fix $\kappa=T$ and the set $\Y_\tau(T)$ contains observations across multiple change points (e.g., distinct distributions), then it is possible that the resulting mixture distribution in $\Y_\tau(T)$ is indistinguishable from the distribution of the observations in $\X_\tau$, even when $\tau$ corresponds to a valid change point. 
We avoid this confounding by allowing $\kappa$ to vary, with minimal computational cost by storing the distances mentioned above.
This modification {to bisection} is similar to that taken in \cite{Olshen:2004}.

\subsection{Hierarchically Estimating Multiple Change Points}\label{finding-multiple-cp}

To estimate multiple change points we iteratively apply the above technique as follows. 
Suppose that $k-1$ change points have been estimated at locations $0<\hat\tau_1<\dots<\hat\tau_{k-1}<T.$
This partitions the observations into $k$ clusters $\widehat C_1, \widehat C_2,\dots, \widehat C_k$,
such that $\widehat C_i = \{Z_{\hat\tau_{i-1}+1},\dots,Z_{\hat\tau_{i}}\}$, in which $\hat\tau_{0} = 0$ and $\hat\tau_{k} = T$.
Given these clusters, we then apply the procedure for finding a single change point to the observations {\it within} each of 
the $k$ clusters. 
Specifically, for the $i$th cluster $\widehat C_i$ denote a proposed change point location as $\hat\tau(i)$ and the associated constant 
$\hat\kappa(i),$ as defined by Equation (\ref{find-tauhat}). 
Now, let
$$i^*=\argmax_{i \in \{1,\ldots,k\}} \ \hat{\cal Q}(\X_{\hat\tau(i)},\Y_{\hat\tau(i)}(\hat\kappa(i));\alpha),$$
in which $\X_{\hat\tau(i)}$ and $\Y_{\hat\tau(i)}(\hat\kappa(i))$ are defined with respect to $\widehat C_i$,
and denote a corresponding test statistic as
\be
\hat q_k = \hat{\cal Q}(\X_{\hat\tau_k},\Y_{\hat\tau_k}(\hat\kappa_k);\alpha),
\ee
in which $\hat\tau_k = \hat\tau(i^*)$ denotes the $k$th estimated change point,  located within cluster $\widehat C_{i^*}$,
and $\hat\kappa_k = \hat\kappa(i^*)$ the corresponding constant. 
This iterative procedure has running time ${\cal O}(kT^2)$, in which $k$ is the unknown number of change points. 

\subsection{Hierarchical Significance Testing}\label{significance-test}

The previous sections have proposed a method for estimating the locations of change points.
We now propose a testing procedure to determine the statistical significance of a change point, conditional on 
previously estimated change points. 
For hierarchical estimation, this test may be used as a stopping criterion for the proposed iterative estimation procedure. 

As above, suppose that $k-1$ change points have been estimated, resulting in $k$ clusters,
and that conditional on $\{\hat\tau_1,\dots,\hat\tau_{k-1}\}$, $\hat\tau_k$ and $\hat q_k$ are the newly proposed change point location and the associated test statistic, respectively. 
Large values of $\hat q_k$ correspond to a significant change in distribution within one of the existing clusters, however, calculating a precise critical value requires 
knowledge of the underlying distributions, which are generally unknown. 
Therefore, we propose a permutation test to determine the significance of $\hat q_k$.

Under the null hypothesis of no additional change points, we conduct a permutation test as follows. 
First, the observations \emph{within} each cluster are permuted to construct a new sequence of length $T$.
Then, we reapply the estimation procedure as described in Sections \ref{finding-a-cp} and \ref{finding-multiple-cp} to the permuted observations. 
This process is repeated and after the $r$th permutation of the observations we record the value of the test statistic
$\hat q_k^{(r)}$.  
 
This permutation test will result in an exact p-value if we consider all possible permutations.
This is not computationally tractable, in general; 
instead we obtain an approximate p-value by performing a sequence of $R$ {\it random} permutations. 
In our implementation we fix the significance level $p_0 \in (0,1)$ of the conditional test, 
as well as the the number of permutations $R$, and the approximate 
p-value is defined as $\#\{r: \hat q_k^{(r)}\ge \hat q_k\}/(R+1)$.
In our analysis we fix $p_0 = 0.05$ and use $R = 499$ permutations for all of our testing.
Determining a suitably large $R$ to obtain an adequate approximation depends on the distribution of the observations, as well as the number and size of clusters. 
As an alternative, a sequential implementation of the random permutations may be implemented with a uniformly bounded resampling risk, see \cite{gandy2009sequential}.


The permutation test may be performed at each stage in the iterative estimation algorithm. The $k$th change point is deemed significant, given $\{\hat\tau_1,\dots,\hat\tau_{k-1}\}$, 
if the approximate p-value is less 
than $p_0,$ and the procedure then estimates an additional location. Otherwise, we are unable to reject the null hypothesis of no additional change points and the algorithm 
terminates. The permutation test may be performed after the E-Divisive procedure reaches a predetermined number of clusters to quickly provide initial estimates. The independent 
calculations of the permuted observations may be performed in parallel to easily reduce computation time. 

Finally, we note that the proposed procedure is not suitable for blockwise stationary observations.  
Such an extension may be possible using the divergence measure in Equation (\ref{DXY}), 
however the resampling procedure must also consider the serial dependence structure of the observations. 

%

\section{Consistency}\label{algorithm-analysis}

We now present results pertaining to the consistency of the estimated change point locations that are returned by the proposed procedure. 
It is assumed throughout that the dimension of the observations is arbitrary, but constant, 
and that the unknown number of change points is also constant. 
Below, we consider the case of a single change point, and demonstrate that we obtain a strongly consistent estimator in a rescaled time setting. 
We then do the same for the more general case of multiple change points.

\subsection{Single Change Point}

In Section \ref{measuring-differences} we have stated that in the case of a single change point, at a given location, the two-sample test is 
statistically consistent against all alternatives. We now show that $\hat\tau$ is a strongly consistent estimator for a single change point location 
within the setting described.

\begin{asm}\label{asm1}
Suppose that we have a heterogeneous sequence of independent observations from two different distributions. 
Specifically, let $\gamma\in(0,1)$ denote the fraction of the observations belonging to one of the distributions, such that $Z_1,\dots,Z_{\lfloor{\gamma T}\rfloor}\sim F_x$ and 
$Z_{\lfloor\gamma T\rfloor+1},\dots,Z_T\sim F_y$ for every sample of size $T$. Let $r=\lfloor\gamma T\rfloor$ and $s=T-r$. 
Also, let $\mu_X^\alpha=E|X-X'|^\alpha$, $\mu_Y^\alpha=E|Y-Y'|^\alpha$, and $\mu_{XY}^\alpha=E|X-Y|^\alpha$, in which
$X, X' \stackrel{iid}{\sim} F_x,$  $Y, Y' \stackrel{iid}{\sim} F_y,$ and $X, X', Y,$ and $Y'$ are mutually independent. 
Further, suppose 
$E(|X|^{\alpha} + |Y|^{\alpha}) < \infty$ for some $\alpha\in(0,2)$; 
hence, $\mu_X^\alpha, \mu_Y^\alpha, \mu_{XY}^\alpha, \mathcal E(X,Y;\alpha)<\infty.$ Finally, let $\{\delta_T\}$ be a sequence of positive 
numbers such that $\delta_T\to 0$ and $T\delta_T\to\infty,$ as $T\to\infty$.
\end{asm}
%
\begin{lem}\label{uniform-conv}
Suppose Assumption \ref{asm1} holds, 
then 
$$\sup_{\gamma\in[\delta_T,1-\delta_T]} \left|\binom{T}{2}^{-1}\sum_{i<j}|Z_i-Z_j|^\alpha\ - \left[ \gamma^2\mu_X^\alpha+(1-\gamma)^2\mu_Y^\alpha+2\gamma(1-\gamma)\mu_{XY}^\alpha \right] \right|
\stackrel{a.s.}{\to}0, \;\; as \;\; T\rightarrow\infty. $$
\end{lem}
\begin{proof}
Let $\epsilon>0$. Define the following disjoint sets: $\Pi_1=\{(i,j): i<j, Z_i,Z_j\sim F_x\}$; $\Pi_2=\{(i,j): Z_i\sim F_x, Z_j\sim F_y\}$; and 
$\Pi_3=\{(i,j): i<j, Z_i,Z_j\sim F_y\}$. By the Strong Law of Large Numbers for $U$-statistics, we have that with probability 1, $\exists N_1\in\mathbb N$ 
such that 
$$\left|\binom{\#\Pi_1}{2}^{-1}\sum_{\Pi_1}|Z_i-Z_j|^\alpha-\mu_X^\alpha\right|<\epsilon$$
whenever $\#\Pi_1>N_1$. By the same argument we can similarly define $N_2,N_3\in\mathbb N$. 
Furthermore, $\exists N_4\in\mathbb N$ such that $\frac{1}{T-1}<\epsilon/2$ for $T>N_4$. 
Let 
$N=N_1\vee N_2\vee N_3\vee N_4$, such that for any $T \delta_T >N$, and every $\gamma\in[\delta_T, 1-\delta_T]$, we have 
$\#\Pi_1=\lfloor\gamma T\rfloor>N_1$, $\#\Pi_2=\lfloor\gamma T\rfloor(T-\lfloor\gamma T\rfloor)>N_2$, 
$\#\Pi_3=(T-\lfloor\gamma T\rfloor)>N_3$, and the quantities $|\frac rT-\gamma|$, 
$|\frac{r-1}{T-1}-\gamma|$, $|\frac sT-(1-\gamma)|$, $|\frac{s-1}{T-1}-(1-\gamma)|$ are each less than $\epsilon$.

Now, considering the nature of the summands, $\frac{2}{T(T-1)}\sum_{\Pi_1}|Z_i-Z_j|^\alpha$ may be rewritten as  
$$\binom{r}{2}^{-1}\left(\frac{r}{T}\right)\left(\frac{r-1}{T-1}\right)\sum_{\Pi_1}|Z_i-Z_j|^\alpha.$$
For $T>N$, we have 
$$P\left(\left|\binom{r}{2}^{-1}\left(\frac{r}{T}\right)\left(\frac{r-1}{T-1}\right)\sum_{\Pi_1}|Z_i-Z_j|^\alpha-\gamma^2\mu_X^\alpha\right|<
\epsilon^3+\epsilon^2(2+3\mu_X^\alpha)+\epsilon\right)=1.$$
The last inequality is obtained from noting that $\left|\frac{r}{T}-\gamma\right|\left|\frac{r-1}{T-1}-\gamma\right|<\epsilon^2$ implies 
$\left|\left(\frac{r}{T}\right)\left(\frac{r-1}{T-1}\right)-\gamma^2\right|<\epsilon^2+2\gamma\epsilon$. Therefore, 
$\left|\left(\frac{r}{T}\right)\left(\frac{r-1}{T-1}\right)-\gamma^2\right|\left|\binom{r}{2}^{-1}\sum_{\Pi_1}|Z_i-Z_j|^\alpha-\mu_X^\alpha\right|
<\epsilon^3+2\gamma\epsilon^2;$ rearranging terms, and using the previous inequality yields 
$$\left|\binom{r}{2}^{-1}\left(\frac{r}{T}\right)\left(\frac{r-1}{T-1}\right)\sum_{\Pi_1}|Z_i-Z_j|^\alpha-\gamma^2\mu_X^\alpha\right|
<\epsilon^3+(2\gamma+(1+2\gamma)\mu_X^\alpha)\epsilon+\gamma^2\epsilon<\epsilon^3+\epsilon^2(2+3\mu_X^\alpha)+\epsilon.$$ 
By applying the same approach, we have similar expressions for both $\frac{2}{T(T-1)}\sum_{\Pi_2}|Z_i-Z_j|^\alpha$ 
and ${\frac{2}{T(T-1)}\sum_{\Pi_3}|Z_i-Z_j|^\alpha}.$
Finally, applying the triangle inequality establishes the claim, since $\epsilon$ is arbitrary.  
\end{proof}
%
%
In order to establish the uniform convergence above, it is assumed that $\gamma$ is bounded away from $0$ and $1$, such that $r \wedge s \to \infty$ as $T\to\infty$. 
In application, we impose a minimum size for each cluster when estimating the location of a change point.
This minimum cluster size should be specified \emph{a priori}; 
in our examples we primarily use 30 as the minimum size, but larger sizes may be needed when $\mathcal E(X,Y;\alpha)$ is relatively small.

\begin{thm}\label{const}
Suppose Assumption \ref{asm1} holds. 
Let $\hat\tau_T$ denote the estimated change point location for a 
sample of size $T$, as defined in Equation (\ref{find-tauhat}), here with $\kappa=T$; i.e., using an unmodified bisection approach. 
Then for $T$ large enough $\gamma\in[\delta_T,1-\delta_T]$, and furthermore, for all $\epsilon>0$
$$P\left(\lim_{T\to\infty}\left|\gamma-\frac{\hat\tau_T}{T}\right|<\epsilon\right)=1.$$
%
\end{thm}
%
\begin{proof}
Let $T$ be such that $\gamma\in[\delta_T,1-\delta_T]$, then 
for any $\tilde\gamma\in[\delta_T,1-\delta_T]$, let $\X_T(\tilde\gamma)=\{Z_1,\dots,Z_{\lfloor\tilde\gamma T\rfloor}\}$ and $\Y_T(\tilde\gamma)=
\{Z_{\lfloor\tilde\gamma T\rfloor+1},\dots,Z_T\}$ for all $T$. Then
\begin{equation}\label{as-sing}
\widehat{\mathcal E}(\X_T(\tilde\gamma),\Y_T(\tilde\gamma);\alpha)\stackrel{a.s.}{\to}\left(\frac{\gamma}{\tilde\gamma}\mathbbm{1}_{\tilde\gamma\ge\gamma}+
\frac{1-\gamma}{1-\tilde\gamma}\mathbbm{1}_{\tilde\gamma<\gamma}\right)^2\mathcal E(X,Y;\alpha)=h(\tilde\gamma;\gamma)\mathcal E(X,Y;\alpha)
\end{equation}
as $T\to\infty$, uniformly in $\tilde\gamma$. The 
maximum of $h(\tilde\gamma;\gamma)$ 
is attained when $\tilde\gamma=\gamma$.
Now, note that $\frac{1}{T}\widehat{\mathcal Q}(\X_T(\tilde\gamma),\Y_T(\tilde\gamma);\alpha)\stackrel{a.s.}{\to}\tilde\gamma(1-\tilde\gamma)h(\tilde\gamma;\gamma)\mathcal E(X,Y;\alpha)$ 
as $T\to\infty$, uniformly in $\tilde\gamma$. Additionally, the maximum value of $\tilde\gamma(1-\tilde\gamma)h(\tilde\gamma;\gamma)$ is also attained when $\tilde\gamma=\gamma$.
Define
$$ \hat\tau_T=\argmax_{\tau \in\left\{{\lceil T\delta_T\rceil},{\lceil T\delta_T\rceil+1},\dots,{\lfloor T(1-\delta_T)\rfloor}\right\}}
\widehat{\mathcal Q}(\X_\tau,\Y_\tau(T);\alpha),$$
and the interval $\displaystyle\hat\Gamma_T=\argmax_{\tilde\gamma\in[\delta_T,1-\delta_T]}\widehat{\mathcal Q}(\X_T(\tilde\gamma),\Y_T(\tilde\gamma);\alpha)$, 
then 
$\frac{\hat\tau_T}{T}\in\hat\Gamma_T$. Since
\begin{equation*}
\frac{1}{T}\widehat{\mathcal Q}\left(\X_T\left({\hat\tau_T}/{T}\right),\Y_T\left({\hat\tau_T}/{T}\right);\alpha\right)>\frac{1}{T}\widehat{\mathcal Q}(\X_T(\gamma),
\Y_T(\gamma);\alpha)-o(1),
\end{equation*}
we have
\begin{equation*}
\frac{1}{T}\widehat{\mathcal Q}(\X_T(\hat\tau_T/T),\Y_T(\hat\tau_T/T);\alpha)\ge\gamma(1-\gamma)h(\gamma;\gamma)\mathcal E(X,Y;\alpha)-o(1),
\end{equation*}
by the almost sure uniform convergence. Letting $\hat\gamma=\hat\tau_T/T$, it follows that
\begin{eqnarray*}
0 \le \gamma(1-\gamma)h(\gamma;\gamma)\mathcal E(X,Y;\alpha)-\hat\gamma(1-\hat\gamma)h(\hat\gamma;\gamma)\mathcal E(X,Y;\alpha)&\le&
\frac{1}{T}\widehat{\mathcal Q}(\X_T(\hat\gamma),\Y(\hat\gamma);\alpha)+o(1)\\
&&-\hat\gamma(1-\hat\gamma)h(\hat\gamma;\gamma)\mathcal E(X,Y;\alpha) \\
&\rightarrow& 0, 
\end{eqnarray*}
\vspace{-12pt}
as $T \rightarrow \infty.$ For every $\epsilon>0$, there exists $\eta$ such that 
$$\tilde\gamma(1-\tilde\gamma)h(\tilde\gamma;\gamma)\mathcal E(X,Y;\alpha)<\gamma(1-\gamma)h(\gamma;\gamma)\mathcal E(X,Y;\alpha)-\eta$$
for all $\tilde\gamma$ with $|\tilde\gamma-\gamma|\ge\epsilon$. Therefore,
\begin{eqnarray*}
P\left(\lim_{T\to\infty}|\hat\gamma_T-\gamma|\ge\epsilon\right)&\le&P\left(\lim_{T\to\infty}
\hat\gamma_T(1-\hat\gamma_T)h(\hat\gamma_T;\gamma)\mathcal E(X,Y;\alpha)<\gamma(1-\gamma)h(\gamma;\gamma)\mathcal E(X,Y;\alpha)-\eta\right)\\
&=&0.   \hspace{11.3cm}  
 \end{eqnarray*}
\end{proof}
Consistency only requires that each cluster's size increase, but not necessarily at the same rate. 
To consider rates of convergence, additional information about the distribution of the estimators, which depends on the unknown distributions of the data, is also necessary.
%

\subsection{Multiple Change Points}\label{multi-cp}

The consistency result presented in \cite{Vostrikova:1981} cannot be applied in this general situation because it assumes that the expectation of 
the observed sequence consists of a piecewise linear function, making it only suitable for estimating change points resulting from breaks in expectation.

\begin{asm}\label{asm2}
Suppose that we have a heterogeneous sequence of independent observations from $k+1$ distributions, denoted $\{F_i\}_{i=0}^{k}$. Specifically, let 
$0=\gamma^{(0)}<\gamma^{(1)}<\cdots<\gamma^{(k)}<\gamma^{(k+1)}=1$. 
Then, for $i=0,1,\dots,k$ we have 
$Z_{\lfloor T\gamma^{(i)}\rfloor+1},\dots,Z_{\lfloor T\gamma^{(i+1)}\rfloor} \stackrel{iid}{\sim} F_i$, such that $F_i\neq F_{i+1}$.
Let $\mu_{ii}^\alpha=E|X_i-X_i'|^\alpha$ and $\mu_{ij}^\alpha=E|X_i-X_j|^\alpha$, in which $X_i,X_i'\stackrel{iid}{\sim}F_i$, independent of $X_j\sim F_j$.
Furthermore, suppose that $\displaystyle\sum_{i=0}^{k}E|X_i|^\alpha<\infty$ for some $\alpha\in(0,2)$;
hence $\mu_{ii}^\alpha$, $\mu_{ij}^\alpha, \mathcal E(X_i,X_j;\alpha)<\infty$, for all $i$ and $j$. Let $\{\delta_T\}$ be a 
sequence of positive numbers such that $\delta_T\to 0$ and $T\delta_T\to\infty,$ as $T\to\infty.$
\end{asm}

Under Assumption \ref{asm2}, analysis of multiple change points can be reduced to the analysis of only two change points. 
For any $i\in\{1,\dots,k-1\}$, consider $\gamma^{(i)}$ and $\gamma^{(i+1)}$. 
The observations $\{Z_j: j\le\lfloor T\gamma^{(i)}\rfloor\}$ can be seen as a random sample from a mixture of distributions $\{F_j: j\le i\},$ denoted here as $F$. Similarly, observations $\{Z_j: j\ge\lfloor T\gamma^{(i+1)}\rfloor+1\}$ are a sample from a mixture of distributions $\{F_j : j>i+1\}$, denoted here as $H$. The remaining observations are distributed according to some distribution $G$. Furthermore, $F \neq G$ and $G \neq H$, if not, we refer to the single change point setting. For notation, we simply consider $\gamma^{(1)}$ and $\gamma^{(2)}$.

Let $X, Y, U$ be random variables such that $X\sim F$, $Y\sim H$, and $U\sim G$.
Consider any $\tilde\gamma$ such that, 
$\gamma^{(1)}\le\tilde\gamma\le\gamma^{(2)}$, then this choice of $\tilde\gamma$ will create two mixture distributions. {One with component 
distributions $F$ and $G$, and the other with component distributions $H$ and $G$. Then the divergence measure in Equation (\ref{DXY}) between these two mixture distributions} 
 is equal to
\begin{equation}\label{mix-meas}
\int_{\mathbb R^d}\bigg|\frac{\gamma^{(1)}}{\tilde\gamma}\phi_x(t)+\left(\frac{\tilde\gamma-\gamma^{(1)}}{\tilde\gamma}\right)\phi_u(t)-
\left(\frac{1-\gamma^{(2)}}{1-\tilde\gamma}\right)\phi_y(t)-\left(\frac{\gamma^{(2)}-\tilde\gamma}{1-\tilde\gamma}\right)
\phi_u(t)\bigg|^2 \, w(t ;\alpha ) \, dt
\end{equation}
%
\begin{lem}\label{right-choice}
Suppose that Assumption \ref{asm2} holds for some $\alpha\in(0,2)$, then the divergence measure in Equation (\ref{mix-meas}) is maximized 
when either $\tilde\gamma=\gamma^{(1)}$ or $\tilde\gamma=\gamma^{(2)}$.
\end{lem}
%
\begin{proof}
Equation (\ref{mix-meas}) can be rewritten as
\begin{equation}\label{mix-meas2}
f(\tilde\gamma)=\int_{\mathbb R^d}\bigg|\frac{\gamma^{(1)}}{\tilde\gamma}[\phi_x(t)-\phi_u(t)]+\frac{1-\gamma^{(2)}}{1-\tilde\gamma}[\phi_u(t)-\phi_y(t)]\bigg|^2\, w(t ;\alpha ) \, dt.
\end{equation}
We then express the above integral as the sum of the following three integrals:
\begin{equation*}
\displaystyle\left(\frac{\gamma^{(1)}}{\tilde\gamma}\right)^2\int_{\mathbb R^d}|\phi_x(t)-\phi_u(t)|^2\, w(t ;\alpha ) \, dt;
\end{equation*}
\begin{equation*}
\displaystyle\frac{2\gamma^{(1)}(1-\gamma^{(2)})}{\gamma(1-\tilde\gamma)}\int_{\mathbb R^d}|\phi_x(t)-\phi_u(t)||\phi_u(t)-\phi_y(t)|\, w(t ;\alpha ) \, dt; \quad \mbox{and}
\end{equation*}
\begin{equation*}
\displaystyle\left(\frac{1-\gamma^{(2)}}{1-\tilde\gamma}\right)^2\int_{\mathbb R^d}|\phi_u(t)-\phi_y(t)|^2\, w(t ;\alpha ) \, dt.
\end{equation*}
Each of these is a strictly convex positive function of $\tilde\gamma$, and therefore so is their sum. Since $\gamma^{(1)}\le\tilde\gamma
\le\gamma^{(2)}$, the maximum value is attained when either $\tilde\gamma=\gamma^{(1)}$ or $\tilde\gamma=\gamma^{(2)}$.
\end{proof}
%
\begin{lem}\label{multi-uniform}
%
Suppose that Assumption \ref{asm2} holds for some $\alpha\in(0,2)$, then 
$$\sup_{\tilde\gamma \in [\gamma^{(1)}, \gamma^{(2)}]} \left|\widehat{\mathcal E}(\X_T(\tilde\gamma),\Y_T(\tilde\gamma);\alpha) -  f(\tilde\gamma) \right|
\stackrel{a.s.}{\to}0, \;\; as \;\; T\rightarrow\infty. $$
\end{lem}
\begin{proof}
Let $p(\tilde\gamma;\gamma)=\frac{\gamma^{(1)}}{\tilde\gamma}$ and $q(\tilde\gamma;\gamma)=\frac{1-\gamma^{(2)}}{1-\tilde\gamma}$. 
Using methods from the proof of Lemma \ref{thm2}, Equation (\ref{mix-meas2}) is equal to
\begin{eqnarray*}
p(\tilde\gamma;\gamma)^2\mathcal E(X,U;\alpha) &+& q(\tilde\gamma;\gamma)^2\mathcal E(Z,U;\alpha) \\
&+& 2pq(\tilde\gamma;\gamma)\left(E|X-U|^\alpha+E|Y-U|^\alpha-E|X-Y|^\alpha-E|U-U'|^\alpha\right).
\end{eqnarray*}
Since $\min\!\!\left(\frac{\gamma^{(1)}}{\gamma^{(2)}},\frac{1-\gamma^{(2)}}{1-\gamma^{(1)}}\right)>0$, by Lemma \ref{uniform-conv} the within distances 
for $X_T(\tilde\gamma)$ and $Y_T(\tilde\gamma)$ converge uniformly to 
$$p(\tilde\gamma;\gamma)^2E|X-X'|^\alpha+(1-p(\tilde\gamma;\gamma))^2E|U-U'|^\alpha+2p(\tilde\gamma;\gamma)(1-p(\tilde\gamma;\gamma))E|X-U|^\alpha \quad \mathrm{and}$$
$$q(\tilde\gamma;\gamma)^2E|Y-Y'|^\alpha+(1-q(\tilde\gamma;\gamma))^2E|U-U'|^\alpha+2q(\tilde\gamma;\gamma)(1-q(\tilde\gamma;\gamma))E|Y-U|^\alpha,$$
respectively. Similarly, it can be shown that the between distance converges uniformly to 
\begin{eqnarray*}
& & pq(\tilde\gamma;\gamma)E|X-Y|^\alpha 
+ p(\tilde\gamma;\gamma)(1-q(\tilde\gamma;\gamma))E|X-U|^\alpha + \\
& & (1-p(\tilde\gamma;\gamma))(1-q(\tilde\gamma;\gamma))E|U-U'|^\alpha 
+ (1-p(\tilde\gamma;\gamma))q(\tilde\gamma;\gamma)E|Y-U|^\alpha.
\end{eqnarray*}
Combining twice the between less the within distances provides the desired quantity.
\end{proof}

Under Assumption \ref{asm2}, for each $i = 0,1,\dots,k,$ there exist distributions $F_i$, $G_i$, and $H_i$ 
such that for $\gamma^{(i)}\le\tilde\gamma\le\gamma^{(i+1)}$, Equation (\ref{mix-meas2}) holds; otherwise $f_i(\tilde\gamma) = 0$. 
By Lemmas \ref{right-choice} and \ref{multi-uniform}, $f_i(\tilde\gamma)$ is maximized when $\tilde\gamma=\gamma^{(i)}$ or 
$\tilde\gamma=\gamma^{(i+1)}$ for $i = 1,2,\dots,k-1$. 
By Theorem \ref{const}, $f_0(\tilde\gamma)$ and $f_k(\tilde\gamma)$ are maximized at $\gamma^{(1)}$ 
and $\gamma^{(k)}$, respectively.
%
\begin{thm}\label{multi-const}
Suppose that Assumption \ref{asm2} holds for some $\alpha\in(0,2)$. 
For $\mathscr A_T\subset (\delta_T, 1-\delta_T)$ and $x\in\mathbb R$,
define $d(x,\mathscr A_T)=\inf\{|x-y|: y\in \mathscr A_T\}$. Additionally, define 
$\displaystyle f(\gamma)=\gamma(1-\gamma)\sum_{i=0}^kf_i(\gamma)$. Let $\hat\tau_T$ be the estimated change point as defined by Equation (\ref{find-tauhat}), 
and $\mathscr A_T=\{y\in[\delta_T,1-\delta_T]: f(y)\ge f(\gamma),\ \forall\gamma\}$. 
Then $d(\hat\tau_T/T,\mathscr A_T)\stackrel{a.s.}{\to} 0$ as $T\to\infty$.
\end{thm}
%
\begin{proof}
First we observe that $\frac{1}{T}\widehat{\mathcal Q}(\X_T(\tilde\gamma),\Y_T(\tilde\gamma);\alpha)\stackrel{a.s.}{\to}f(\tilde\gamma)$ 
as $T\to\infty$, uniformly in $\tilde\gamma$ by Lemma \ref{multi-uniform}. Also, for each $i$, $\tilde\gamma(1-\tilde\gamma)f_i(\tilde\gamma)$ is a strictly convex function. 
Therefore, for $T$ large enough, $\delta_T<\gamma^{(1)}$ and $\gamma^{(k)}<1-\delta_T$, so that $\mathscr A_T\neq\emptyset$.  
Since 
$\tilde\gamma(1-\tilde\gamma)f_i(\tilde\gamma)$ is continuously differentiable and strictly convex, there exists a $c_i>0$, such that for any 
$\tilde\gamma_1,\tilde\gamma_2\in[\gamma^{(i)},\gamma^{(i+1)}]$,
\begin{equation}\label{convex-note}
|\tilde\gamma_1(1-\tilde\gamma_1)f_i(\tilde\gamma_1)-\tilde\gamma_2(1-\tilde\gamma_2)f_i(\tilde\gamma_2)|>c_i|\tilde\gamma_1-\tilde\gamma_2|+o(|\tilde\gamma_1-\tilde\gamma_2|).
\end{equation}
Let $\epsilon>0$. By Equation (\ref{convex-note}), there exists $\eta(\epsilon)>0$ such that if 
$d(\tilde\gamma,\mathscr A_T)>\eta(\epsilon)$, then $|f(\tilde\gamma)-f(x)|>\epsilon$, for all $x\in\mathscr A_T$. 
Now, let $\hat\gamma_T={\hat\tau_T}/{T}$ and $\gamma^* = \argmin_{x \in {\mathscr A_T}} |\hat\gamma_T-x|,$
then
$$f(\hat\gamma_T)+\frac{\epsilon}{2}>\frac{1}{T}\widehat{\mathcal Q}(\X_T(\hat\gamma_T),\Y_T(\hat\gamma_T);\alpha)\ge 
\frac{1}{T}\widehat{\mathcal Q}(\X_T(\gamma^*),\Y_T(\gamma^*);\alpha)>f(\gamma^*)-\frac{\epsilon}{2},$$
with probability 1. Combining the first and last terms in the above expression provides us with $f(\gamma^*)-f(\hat\gamma_T)<\epsilon$. 
Therefore, $\displaystyle P\!\left(\lim_{T\to\infty}d\!\left({\hat\tau_T}/{T},\mathscr A_T\right)\le\eta(\epsilon)\right)=1$, and since $\epsilon$ was arbitrary, we have established the claim.
\end{proof}

Repeated application of Theorem \ref{multi-const} shows that as $T\to\infty$, the first $k$ estimated change points will converge to the true change point locations in the manner described above. 
With a fixed significance level $p_0$, all of these will be identified. However, the testing procedure may identify additional spurious change points, the number of which is distributed geometric, by construction.


\section{Simulation Study}\label{simulations}

In this section we present simulation results from the E-Divisive procedure using various univariate and multivariate distributions.
We compare performance with the MultiRank procedure \citep[see,][]{Fong:2011},
which is based on a generalization of a Wilcoxon/Mann-Whitney (marginal) rank based approach, {the parametric Pruned Exact Linear Time (PELT) procedure} \citep{Killick:2012}{, and 
the nonparametric Kernel Change Point (KCP) procedure} \citep{Arlot:2012}.
Each simulation applies 
these methods to a set of 1,000 independent sequences with two change points, and 
computes the average Rand index \citep{Fowlkes:1983,Hubert:1985}, defined below, and approximate standard errors. 
All computation was completed using the statistical software {\tt R} \citep{R:citation}, using the {\tt ecp} package \citep[see][]{2013arXiv1309.3295J}.

Throughout this section the E-Divisive procedure was implemented with $\alpha=1;$ results for $\alpha=0.5, 1.5$ were similar, and within the margin of error. We used $R = 499$ iterations when performing the permutation test, 
which was conducted at the marginal $p_0 = 0.05$ significance level. Furthermore, we set the minimum cluster size for the E-Divisive procedure to 30.   
The MultiRank {and KCP} procedure require upper limit{s} on the number of change points, {these were} set to 
$\frac{T}{30}-1$, in which $T$ is the length of the sequence.

\subsection{Comparing Sets of Change Point Estimates}\label{rand}

To measure the performance of a particular method we calculate the Rand index \citep{Rand:1971} as well as Morey and Agresti's  Adjusted Rand index \citep{Morey:1984}. 
These indices represent a measure of similarity between two different partitions of the same observations. 
The first is most suitable for comparing an estimated set of change points to a baseline or known set of locations, while the second is tailored to compare two sets of estimated change points. In both cases, the number of change points in each set need not be equal. 

Suppose that the two clusterings of $T$ observations are given by $U = \{U_1,\dots, U_a\}$ and $V = \{V_1,\dots, V_b\}$, with $a$ and $b$ clusters, respectively. 
For these two clusterings, the Rand index is calculated by noting the relative cluster membership for all \emph{pairs} of observations. 
Consider the pairs of observation that fall into one of the following two sets:
$\{A\}$ pairs of observation in same cluster under $U$ and in same cluster under $V$;
$\{B\}$ pairs of observation in different cluster under $U$ and in different cluster under $V$. 
%
%
%
Let $\#A$ and $\#B$ denote the number of pairs of observation in each of these two sets, respectively. 
The Rand index is then defined as
$$\mbox{Rand } = \frac{\#A+\#B}{\binom{T}{2}}.$$
%

One shortcoming of the Rand index is that it is difficult to compare two different estimated sets of clusterings, since it does not measure the departure from a given baseline model. 
As mentioned in \cite{Hubert:1985}, the Rand index, as well as other similarity indices, are not adjusted for chance
(e.g., the index does not take on a constant value when comparing two random clusterings) for a given model of randomness. 
A common model of randomness, used in \cite{Hubert:1985} and \cite{Fowlkes:1983}, is the hypergeometric model,
which conditions on both the number of clusters and their sizes. 
Under this model, the adjustment for chance requires the expected index 
value and its maximum value. 
An Adjusted Rand index is then defined as
$$\mbox{Adjusted Rand } = \frac{\mbox{Rand} - \mbox{Expected Rand}}{1 - \mbox{Expected Rand}},$$
in which 1 corresponds to the maximum Rand index value.

\subsection{Univariate Analysis}\label{univariate}
In this section we compare the simulation performance of the E-Divisive, MultiRank, and the PELT algorithms on various univariate sequences. 
Within these simulations, 
we attempt to identify change points that resulted because of a distributional change in mean, variance, or tail shape. 
The magnitude of these respective changes was also varied, as shown in Table \ref{sim-all}. 

{For detecting changes in mean and variance, the E-Divisive procedure compares favorably with the parametric PELT procedure. Since the PELT 
procedure is specifically designed to only identify changes in mean or variance, we compare the E-Divisive and MultiRank procedures when considering 
changes in tail shape. 
The sample size was also varied $T = 150, 300, 600$, while the three clusters maintained equal sizes of $T/3,$ with distributions $N(0,1), G, N(0,1)$, respectively.  
We note that the Rand index values for the E-Divisive procedure tend towards $1$ as the sample size increases.
This follows from the consistency established in Theorem \ref{multi-const}. 

\begin{table}[ht]
\begin{center}
{ \footnotesize 
\begin{tabular}{|c|| c| l| l|| c| l| l|| c| c| c|}
\hline
  \multicolumn{1}{|c||}{}  & \multicolumn{3}{c||}{{\small\textbf{Change in Mean}}} & \multicolumn{3}{c||}{{\small\textbf{Change in Variance}}} & \multicolumn{3}{c|}{{\small\textbf{Change in Tail}}} \\
  \hline
  $T$ & $\mu$ & \multicolumn{1}{c|}{E-Divisive} &  \multicolumn{1}{c||}{PELT} & $\sigma^2$ & \multicolumn{1}{c|}{E-Divisive} &  \multicolumn{1}{c||}{PELT} & $\nu$ & E-Divisive & MultiRank\\
  \noalign{\hrule height 2pt}
\multirow{3}{*}{150}  &   1   & $0.950_{0.001}$ & $0.945_{0.002}$ &           2   & $0.907_{0.003}$ &$0.935_{0.002}$ &   16  & $0.835_{0.017}$&$0.631_{0.005}$ \\
\cline{2-10}          &   2   & $0.992_{4.6\e{-4}}$ & $0.990_{4.1\e{-4}}$  &   5   & $0.973_{0.001}$ &$0.987_{4.7\e{-4}}$ &   8   & $0.836_{0.020}$&$0.648_{0.005}$  \\
\cline{2-10}          &   4   & $1.000_{3.7\e{-5}}$ & $0.999_{9.3\e{-5}}$ &   10  & $0.987_{7.1\e{-4}}$ &$0.994_{2.7\e{-4}}$ &   2   & $0.841_{0.011}$&$0.674_{0.004}$ \\
\hline\hline 
\multirow{3}{*}{300}  &   1   & $0.972_{9.1\e{-4}}$   & $0.973_{8.9\e{-4}}$  &   2    & $0.929_{0.003}$&$0.968_{0.001}$ &   16  & $0.791_{0.015}$&$0.624_{0.007}$ \\
\cline{2-10}          &   2   & $0.996_{2.2\e{-4}}$   & $0.994_{2.3\e{-4}}$  &   5    & $0.990_{5.1\e{-4}}$&$0.994_{2.1\e{-4}}$ &   8   & $0.729_{0.018}$&$0.639_{0.006}$\\
\cline{2-10}          &   4   & $1.000_{1.0\e{-5}}$   & $1.000_{4.5\e{-5}}$  &   10   & $0.994_{3.2\e{-4}}$&$0.998_{1.2\e{-4}}$ &   2   & $0.815_{0.006}$&$0.682_{0.006}$\\
\hline\hline  
\multirow{3}{*}{600}  &   1   & $0.987_{1.5\e{-5}}$ & $0.987_{4.1\e{-4}}$ &   2   & $0.968_{0.001}$ &$0.984_{5.1\e{-4}}$ &   16  & $0.735_{0.019}$&$0.647_{0.016}$ \\
\cline{2-10}          &   2   & $0.998_{3.9\e{-6}}$ & $0.997_{1.1\e{-4}}$ &   5   & $0.995_{2.2\e{-4}}$ &$0.997_{1.1\e{-4}}$ &   8   & $0.743_{0.025}$&$0.632_{0.016}$\\
\cline{2-10}          &   4   & $1.000_{3.1\e{-7}}$ & $1.000_{2.3\e{-5}}$ &   10  & $0.998_{1.5\e{-4}}$ &$0.999_{6.4\e{-5}}$ &   2   & $0.817_{0.006}$&$0.708_{0.010}$\\
\hline   
\end{tabular} }
\caption{\label{sim-all} Average Rand index and approximate standard errors from 1,000 simulations for the E-Divisive, PELT and MultiRank methods. 
Each sample has $T = 150, 300 \;\mathrm{or}\; 600$ observations, consisting of three equally sized clusters, with distributions $N(0,1), G, N(0,1)$, respectively. 
For changes in mean $G = N(\mu,1)$, with $\mu = 1, 2,$ and $4$;
for changes in variance $G = N(0,\sigma^2)$, with $\sigma^2 = 2, 5,$ and $10$;
and for changes in tail shape $G = t_{\nu}(0,1)$, with $\nu = 16, 8,$ and $2$.
}
\end{center}
\end{table}

\vspace{-12pt}
\subsection{Multivariate Analysis}\label{multivariate}

We next compare the results of running 
the E-Divisive, KCP and MultiRank methods on bivariate observations. 
In these simulations the distributional differences are either a change in mean or correlation. 
The results of these simulations can be found in Table \ref{sim-bi-all}. 
Let $N_2(\bm{\mu},\Sigma_{\rho})$ denote the bivariate normal distribution with mean vector $\bm{\mu} = (\mu, \mu)'$ and covariance matrix 
$\Sigma_{\rho}=\begin{pmatrix}1&\rho\\\rho&1\end{pmatrix}$ for $\rho \in (-1,1)$, or simply the identity $I$ for $\rho = 0$.
We use the same setup as in the previous section, with observations from $N_2(\bm{0},I), G, N_2(\bm{0},I)$ distributions, respectively.  

For a simultaneous change in mean, with $G = N_2(\bm{\mu},I)$, all methods performed similarly. 
When detecting changes in correlation, with $G = N_2(\bm{0},\Sigma_{\rho}),$ the KCP approach performed best when the sample size was sufficiently 
large for it to detect any changes. However, its computational time was about three times longer than E-Divisive, for these simulations. 
The MultiRank method was not reliable for detecting changes in correlation. 


\begin{table}[ht]
\begin{center}
{ \small 
\begin{tabular}{|c|| c| c| c| c|| c| c| c| c|}
\hline
  \multicolumn{1}{|c||}{}  & \multicolumn{4}{c||}{{\small\textbf{Change in Mean}}} & \multicolumn{4}{c|}{{\small\textbf{Change in Correlation}}} \\
    \hline
      $T$ & $\mu$ & E-Divisive & KCP& MultiRank& $\rho$ & E-Divisive& KCP& MultiRank\\
        \noalign{\hrule height 2pt}
	\multirow{3}{*}{300}  &   1   & $0.987_{4.7\e{-4}}$ & $0.985_{6.6\e{-4}}$ & $0.983_{4.8\e{-4}}$&  0.5   & $0.712_{0.018}$ & $0.331_{N/A \;}$ & $0.670_{0.006}$\\
	\cline{2-9}                &     2 & $0.992_{8.9\e{-5}}$& $0.998_{1.1\e{-4}}$& $0.991_{1.1\e{-4}}$& 0.7   & $0.758_{0.021}$ & $0.331_{N/A \;}$ & $0.723_{0.004}$\\
	\cline{2-9}                &    3   & $1.000_{1.3\e{-5}}$& $1.000_{3.9\e{-5}}$& $0.991_{5.1\e{-5}}$&  0.9   & $0.769_{0.017}$ & $0.331_{N/A \;}$ & $0.748_{0.002}$\\
	\hline\hline   
	\multirow{3}{*}{600}  &   1   & $0.994_{2.2\e{-4}}$ & $0.993_{2.3\e{-4}}$ & $0.992_{2.1\e{-4}}$&  0.5   & $0.652_{0.022}$ & $0.331_{N/A \;}$ & $0.712_{0.011}$\\
	\cline{2-9}                &     2 & $1.000_{4.3\e{-5}}$ & $0.999_{5.2\e{-5}}$ & $0.995_{5.3\e{-5}}$&  0.7   & $0.650_{0.017}$ & $0.848_{0.073}$& $0.741_{0.006}$\\
	\cline{2-9}                &    3   & $1.000_{3.3\e{-6}}$ & $1.000_{2.2\e{-5}}$ & $0.996_{2.7\e{-5}}$&  0.9   & $0.806_{0.019}$ &  
$0.987_{0.001}$& $0.748_{0.002}$\\
	\hline\hline   
	\multirow{3}{*}{900}  &   1   & $0.996_{1.6\e{-4}}$ & $0.995_{1.6\e{-4}}$ & $0.995_{1.3\e{-4}}$&  0.5   & $0.658_{0.024}$ & $0.778_{0.048}$ & $0.666_{0.044}$\\
	\cline{2-9}                &     2 &$1.000_{3.0\e{-5}}$ & $0.999_{4.0\e{-5}}$ & $0.997_{3.5\e{-5}}$&  0.7   & $0.633_{0.022}$ & $0.974_{0.002}$ & $0.764_{0.021}$\\
	\cline{2-9}                &    3   & $1.000_{5.2\e{-6}}$ & $1.000_{1.4\e{-5}}$ & $0.997_{1.8\e{-5}}$&  0.9   & $0.958_{0.004}$ & $0.992_{0.004}$ & $0.741_{0.006}$\\
	\hline
	\end{tabular} }
	\caption{\label{sim-bi-all}
	Average Rand index and approximate standard errors from 1,000 simulations for the E-Divisive, MCP and MultiRank methods. 
	Each sample has $T = 300, 600 \;\mathrm{or}\; 900$ observations, consisting of three equally sized clusters, with distributions $N_2(\bm{0},I), G, N_2(\bm{0},I)$, respectively. 
	For changes in mean $G = N_2(\bm{\mu},I)$, with $\bm{\mu} = (1,1)', (2,2)',$ and $(3,3)'$;
	for changes in correlation $G = N(\bm{0},\Sigma_{\rho})$, in which the diagonal elements of $\Sigma_{\rho}$ are $1$ and the off-diagonal are $\rho$, with $\rho = 0.5, 0.7,$ and $0.9$.
	}
	\end{center}
	\end{table}

The final multivariate simulation examines the performance of the E-Divisive method as the dimension of the data increases. 
In this simulation we consider two scenarios.
\emph{With noise}: in which added components are independent, and do not have a change point.  
\emph{No noise}: in which the added dimensions are correlated, and all marginal and joint distributions have common change point locations. 
The setting is similar to above; each sample of $T = 300, 600,$ or $900$ observations consist of three equally sized clusters, with distributions $N_d(\bm{0},I), G, N_d(\bm{0},I)$, respectively, in which $d$ denotes the dimension, for which we consider $d = 2, 5$ or $9$. 

For the no noise case, we consider $G = N_d(\bm{0},\Sigma_{0.9})$,
in which the diagonal elements of $\Sigma_{0.9}$ are $1$ and the off-diagonal elements are $0.9$.
For the with noise case, we consider $G = N_d(\bm{0},\Sigma_{0.9}^{noise})$,
in which the diagonal elements of $\Sigma_{0.9}^{noise}$ are $1$ and \emph{only} the $(1,2)$ and $(2,1)$ elements are $0.9$, the others are zero, such that a change in distribution occurs in the correlation of only the first two components. 
The results are shown in Table \ref{sim-cordim}. 
The performance of the E-Divisive method improves with increasing dimension when all components of the observed vectors are related, i.e.,\ no noise, even when the number of observations $T$ is fixed. However, the opposite is true when the additional 
components are independent with no change points. 
We conjecture that our method performs better when there are simultaneous changes within the components, and in the presence of noise, dimension reduction may be necessary to obtain comparable performance.

\begin{table}[ht]
\begin{center}
{ \small 
\begin{tabular}{|c| c| l| c|}
  \hline
  $T$ & $d$ & \multicolumn{1}{c|}{{\small\textbf{{No} Noise}}} & \multicolumn{1}{c|}{{\small\textbf{{With} Noise}}}\\ 
  \noalign{\hrule height 2pt}
\multirow{3}{*}{300}  &   2 & $0.723_{0.019}$ & $0.751_{0.018}$ \\
\cline{2-4}
   &    5 & $0.909_{0.010}$&  $0.706_{0.019}$ \\
\cline{2-4}
   &    9  &  $0.967_{0.003}$ & $0.710_{0.026}$ \\
\hline\hline
\multirow{3}{*}{600}  &   2   & $0.930_{0.018}$ & $0.822_{0.019}$ \\
\cline{2-4}
   &   5 &  $0.994_{5.4\e{-4}}$ & $0.653_{0.023}$ \\
\cline{2-4}
   &    9  & $0.997_{3.3\e{-4}}$ & $0.616_{0.021}$ \\
\hline\hline
\multirow{3}{*}{900}  &   2   & $0.967_{0.003}$ & $0.966_{0.003}$ \\
\cline{2-4}
   &     5 &  $0.998_{1.8\e{-4}}$  & $0.642_{0.018}$ \\
\cline{2-4}
   &    9  & $0.999_{1.0\e{-4}}$  & $0.645_{0.021}$ \\
\hline
\end{tabular} }
\caption{ \label{sim-cordim}
Average Rand index and approximate standard errors from 1,000 simulations for the E-Divisive method.
Each sample has $T = 300, 600 \;\mathrm{or}\; 900$ observations, consisting of three equally sized clusters, with distributions $N_d(\bm{0},I), G, N_d(\bm{0},I)$, respectively, in which $d = 2, 5$ or $9$ denotes the dimension.
For the no noise case, $G = N_d(\bm{0},\Sigma_{0.9})$,
in which the diagonal elements of $\Sigma_{0.9}$ are $1$ and the off-diagonal are $0.9$.
For the with noise case, $G = N_d(\bm{0},\Sigma_{0.9}^{noise})$,
in which the diagonal elements of $\Sigma_{0.9}^{noise}$ are $1$ and \emph{only} the $(1,2)$ and $(2,1)$ elements are $0.9$, the others are zero.
}
\end{center}
\end{table}

\vspace{-24pt}
\section{Applications}\label{real-data}
We now present results from applying the proposed E-Divisive procedure, and others, to genetics and financial datasets. 

\subsection{Genetics Data}\label{genetics}

We first consider the genome data from \cite{Vert:2011}. Genome samples for 57 individuals with a bladder tumor are scanned for variations 
in DNA copy number using array comparative genomic hybridization (aCGH). The relative hybridization intensity with respect to a normal genome 
reference signal is recorded. These observations were normalized so that the modal ratio is zero on a logarithmic scale. 

The approach in  \cite{Vert:2011} assumes that each sequence is constant between change points, with additive noise. Thus, this approach is 
primarily concerned with finding a distributional change in the mean. In order to directly apply the procedures we first account for missing 
values in the data; for simplicity, we imputed the missing values as the average of their neighboring values. 
We removed all series that had more than 7\% of values missing; leaving genome samples of 43 individuals for analysis.

When applied to the 43-dimension joint series of individuals, the MultiRank algorithm found 43 change points, while the E-Divisive algorithm found 97 change points, using 
$\alpha = 1$, a minimum cluster size of 10 observations, $R = 499$ permutations and $p_0 = 0.05$ in our significance testing.
%
%
Estimated change point locations, for individual 10, under four methods 
are shown in Figure \ref{fig:genetic}.
MultiRank estimated 17 change points, with adjusted Rand values of 0.572 (Kernel CP), 0.631 (PELT), 0.677 (E-Divisive), respectively.
KCPA estimated 41 change points, with adjusted Rand values of 0.678 (PELT), 0.658 (E-Divisive), respectively. 
PELT estimated 47 change points, with adjusted Rand value of 0.853 (E-Divisive), 
and E-Divisive estimated 35 change points. 

\begin{figure}[ht]
	\centering
	\includegraphics[scale = 0.48]{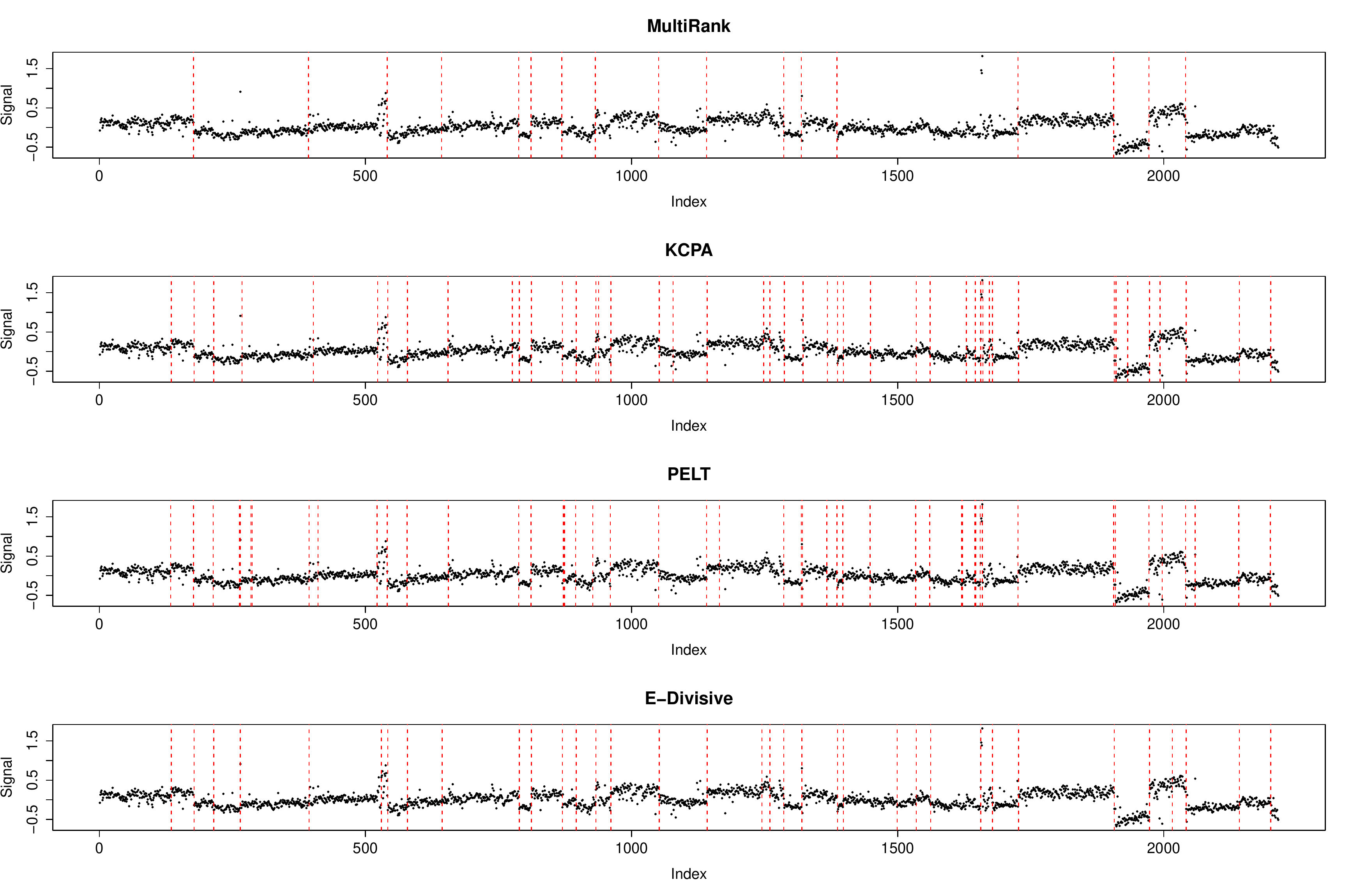} 
	\caption{The normalized relative aCGH signal for the tenth individual with a bladder tumor;
 	the estimated change point locations for the MultiRank, KCPA, PELT and E-Divisive methods
	are indicated by the dashed vertical lines. 
}
	\label{fig:genetic}
\end{figure}

\subsection{Financial Data}\label{financial}
Here we apply the E-Divisive algorithm to the 262 monthly log returns for Cisco Systems Inc. stock, an industry leader in the design and manufacturing 
of networks, from April 1990 through January 2012. In our analysis we specified $\alpha = 1$, a minimum cluster size of 30 observations, and used $R = 499$ permutations 
with a level of $p_0 = 0.05$ in our significance testing. We estimated two significant change points, both with approximate p-values below 0.03. 
The series is shown in Figure \ref{fin1} with vertical lines to denote the estimated change point locations at April 2000 and October 2002.
\begin{figure}[ht]
	\centering
	\includegraphics[width=\columnwidth]{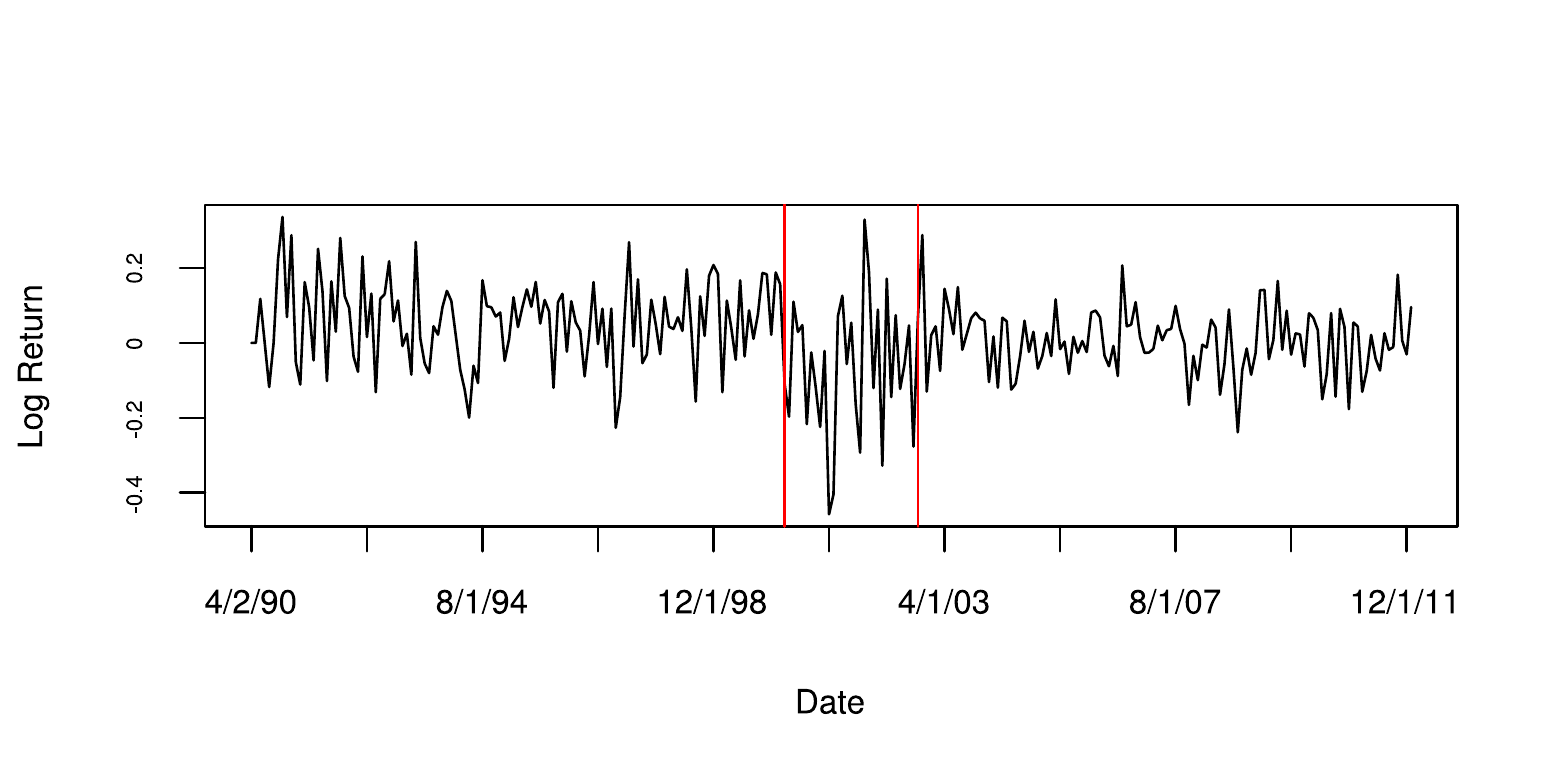}
	\caption{Monthly log returns for Cisco Systems Inc.\ stock, from April 1990 through January 2012;
	the E-Divisive procedure estimates significant changes in distribution at the vertical lines April 2000 and October 2002.
	}
	\label{fin1}
\end{figure}

The change point in April of 2000 corresponds to the company's acquisition of Pirelli Optical Systems to counter rising competitors Nortel and 
Lucent. The acquisition allowed Cisco to provide its customers with lower network costs and a more complete network infrastructure. The 
October 2002 change point represents the end of a period of highly aggressive ventures in emerging markets, during which Cisco was chosen 
to develop a multi-billion dollar network for Shanghai, which became China's largest urban communications network.

Figure \ref{fin2} shows distributional comparisons between the three time periods. Quantile-quantile plots between adjacent time periods are shown in 
the first two plots and kernel density estimates for each of the three periods are shown in the third plot. Included with the kernel density estimates 
are  95\% point-wise confidence bands, which were created by applying a bootstrap procedure to each of the three time periods. The second time period is 
relatively more volatile and skewed than either of its neighboring time periods. 

\begin{figure}[ht]
	\centering
	\includegraphics[width=\columnwidth]{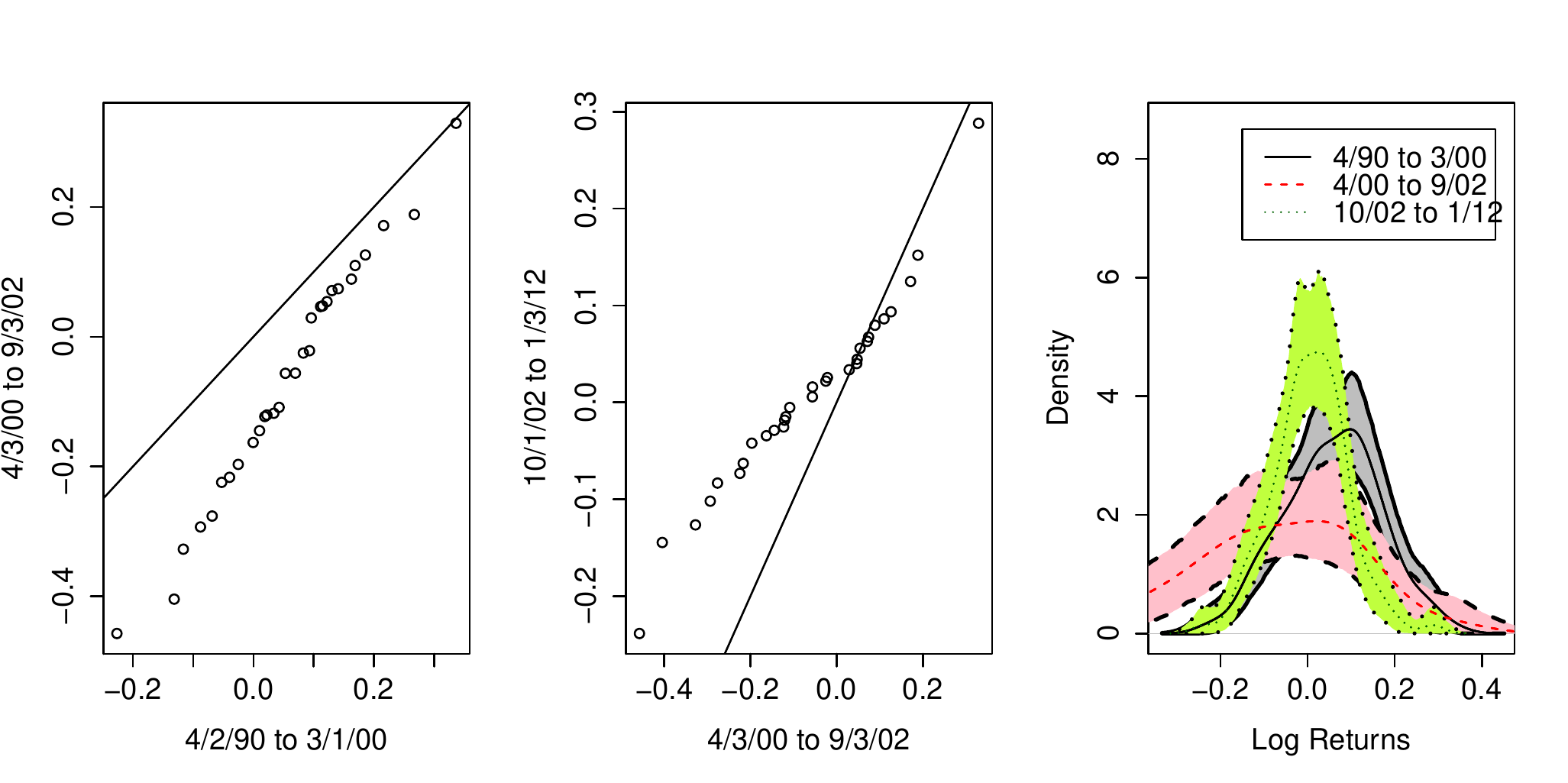}
	\caption{Distributional comparisons between the estimated change points from the E-Divisive procedure: (a,b) quantile-quantile plots between 
adjacent time periods; and (c) kernel density estimates for each period with 95\% confidence bands.}
	\label{fin2}
\end{figure}

To graphically support the assumption of independent observations 
within clusters, Figure \ref{fin3} shows several lags of the sample auto-correlation function (ACF) for the returns (top row) and the squared returns 
(bottom row), for the entire period (first column) and each sub-period (later columns). The dashed horizontal lines represent approximate 95\% confidence intervals about zero, suggesting that the lagged 
correlation statistics are not significant. Within sub-periods there is no significant serial correlation or conditional heteroskedasticity. 
Although there appears to be minor serial dependence when studying the entire series, this is an artifact of the distributional changes over time.

\begin{figure}[!ht]   
	\centering
	\includegraphics{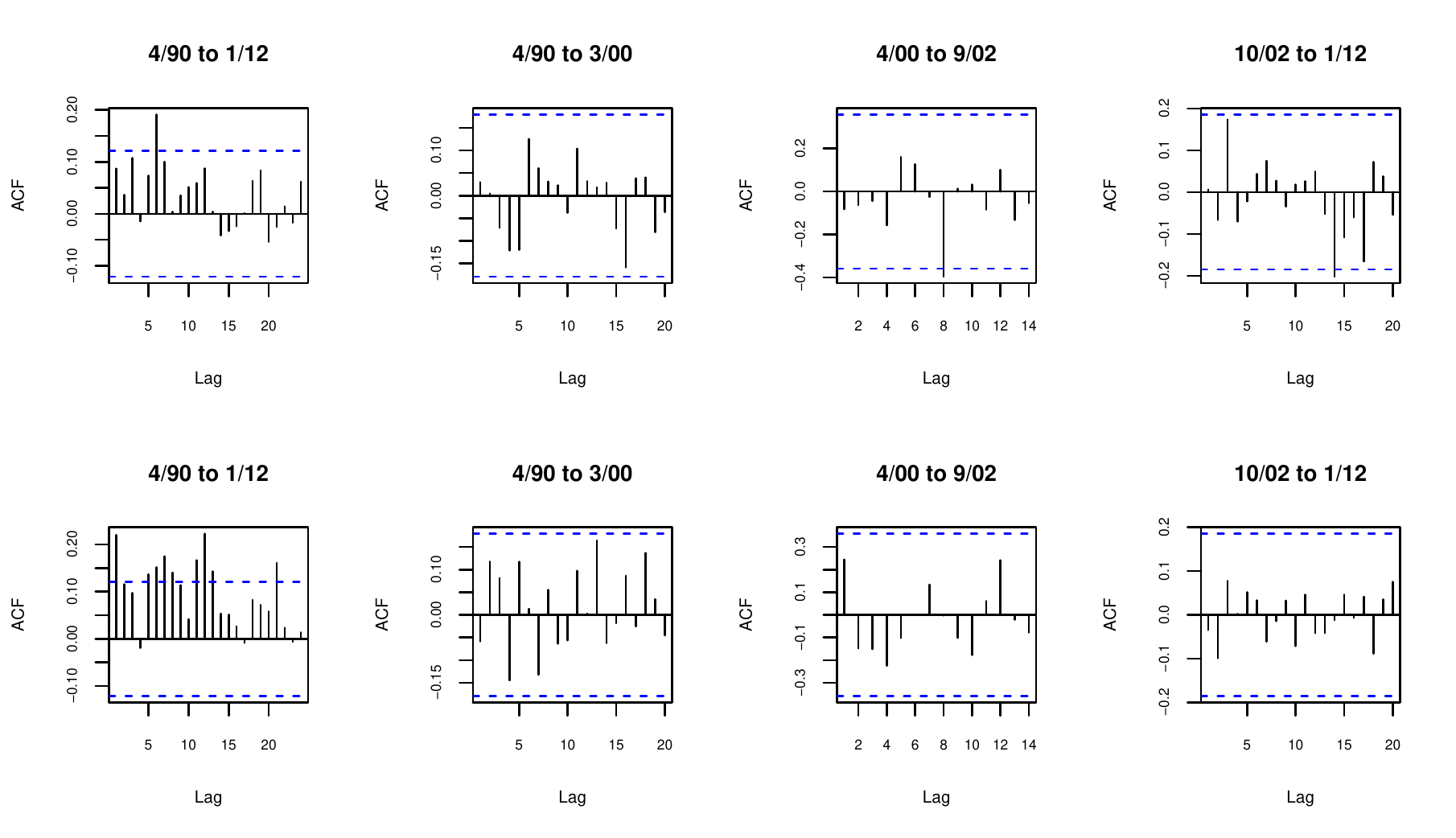}
	\caption{Sample auto-correlation function for the returns (top row) and the squared returns (bottom row), for the entire period (first column) and each estimated sub-period (later columns). The dashed horizontal lines represent approximate 95\% confidence intervals about zero.}
	\label{fin3}
\end{figure}

\section{An Agglomerative Algorithm}\label{agglomerative}

Our hierarchical approach up to this point has only considered the use of a divisive algorithm. 
However, its practical computational effort limits the length of the series that it can analyze without pre-partitioning the series for segment-wise analysis. 
In this section we present an agglomerative approach that has practical advantages over the E-Divisive method. 
Even though this method also has computational time that is quadratic in the length of the series, in practice it runs much faster than the E-Divisive approach. This reduction is accomplished 
by only considering a relatively small subset of possible change point locations; 
a similar restriction to the E-Divisive approach does not result in any computational savings.

\subsection{Overview}
Suppose the sequence of observations $Z_1,Z_2,\dots,Z_T$ are independent, each with finite $\alpha{th}$ absolute moment, for some $\alpha\in(0,2)$. Unlike most general purpose agglomerative clustering algorithms, the proposed procedure will preserve the time ordering of the observations. 
The number of change points will be estimated by the maximization of a goodness-of-fit statistic.

Suppose that we are initially provided a clustering ${\cal C}=\{C_1,C_2,\dots,C_n\}$ of $n$ clusters. These clusters need not consist of a single observation. We then impose 
the following restriction on which clusters are allowed to be merged. Suppose that $C_i=\{Z_k,Z_{k+1},\dots,Z_{k+t}\}$ and $C_j=\{Z_{\ell},Z_{\ell+1},\dots,Z_{\ell+s}\}$. 
To preserve the time ordering, we allow $C_i$ and $C_j$ to merge if either $k+t+1=\ell$ or $\ell+s+1=k$, that is, if $C_i$ and $C_j$ are adjacent. 

To identify which adjacent pair of clusters to merge we use a goodness-of-fit statistic, defined below. 
We greedily optimize this statistic by merging the pair of adjacent clusters that results in either the largest increase or smallest decrease of the statistic's value. 
This process is repeated, recording the goodness-of-fit statistic at each step, until all observations belong to a single cluster. 
Finally, 
the estimated number of change points is estimated by the clustering that maximizes the goodness-of-fit statistic
over the entire merging sequence.

\subsection{Goodness-of-Fit}
The goodness-of-fit statistic we employ is the between-within distance among adjacent clusters.
Suppose that ${\cal C}=\{C_1,
C_2,\dots,C_n\}$, then 
\be\label{gof}
 \widehat{\cal S}_n({\cal C}; \alpha) = \sum_{i=1}^{n-1}\widehat{\cal Q}(C_i,C_{i+1}; \alpha),
\ee
in which $C_i$ and $C_{i+1}$ are adjacent, arranged by relabeling the clusters as necessary,
and $\widehat{\cal Q}$ is defined analogous to Equation (\ref{eqn6}). 

Initialization of the merging sequence $\{ \widehat{\cal S}_k : k = n, \ldots, 2\}$ is performed by calculating $\widehat{\cal Q}$ for \emph{all} pairs of clusters, similar to any agglomerative algorithm. 
We additionally note that once a pair of clusters has been merged, the statistic $\widehat{\cal S}_k$ can be updated to $\widehat{\cal S}_{k-1}$ in ${\cal O}(1)$; hence, the overall complexity of this approach is ${\cal O}(T^2)$.

\subsection{Toronto EMS Data}
In this section we apply the agglomerative algorithm to a spatio-temporal point process dataset. 
Data was collected during 2007 in the city of Toronto for all high priority emergency medical services (EMS) that required at least one ambulance. 
For each of these events a time rounded to the nearest second and a spatial location latitude and longitude were recorded. 
The hourly city-wide emergency event arrival rate was modeled in \cite{MaMcWoHe2011};
 exploratory analysis immediately reveals that the spatial distribution also changes with time. 
This is largely driven by the relative changes in population density as individuals move throughout the city.  

After removing data from holidays and special events, we found significant distributional changes across the course of a week, but little variation from week to week. 
Here we investigate the intra-week changes by pooling all of the approximately 200,000 events from 2007 into a single weekly period, in which time indicates seconds since midnight Saturday.
 Because of the large number of observations, we initialize the agglomerative algorithm by first partitioning the week into 672 equally spaced 15 minute periods. 

The results from running the algorithm with $\alpha = 1$ are shown in the top of Figure \ref{EMS}.
The goodness-of-fit measure in Equation (\ref{gof}) was maximized at 31 change points. 
The estimated change point locations occur everyday, primarily in the evening. 
Several changes occur after little duration, indicating times when the spatial distribution is quickly changing. 
Density estimates from observation in three adjacent cluster periods are shown, on the square-root scale, in the bottom of Figure \ref{EMS}. We note a persistently large density in the downtown region and various shape changes in the outlying regions.  

\begin{figure}[ht]
	\centering
\includegraphics[scale=.49]{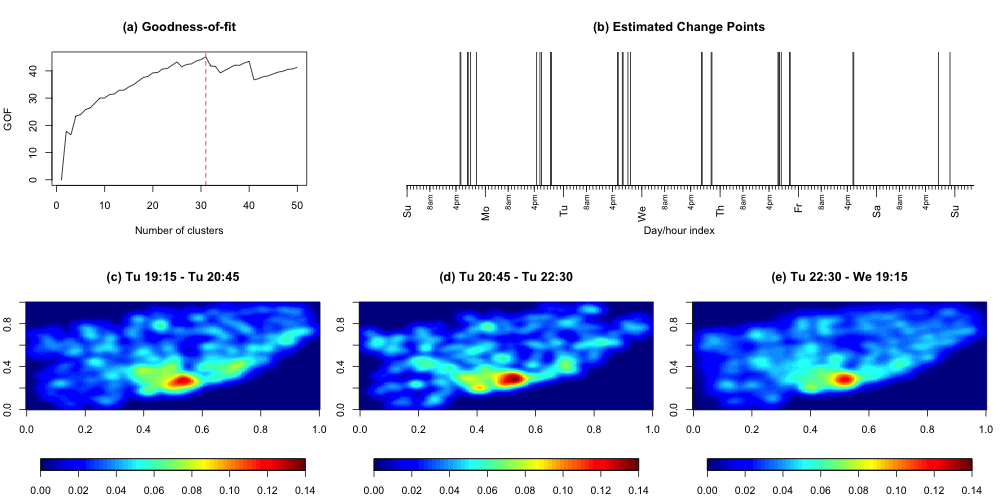}
	\caption{Results from application of the proposed agglomerative algorithm on the Toronto EMS ambulance data:
	(a) the goodness-of-fit measure of Equation (\ref{gof});
	(b) the 31 estimated change point locations;
	and spatial density estimates, on the square-root scale, from observation in three adjacent cluster periods 
	(c) Tuesday 19:15 - 20:45,
	(d) Tuesday 20:45 - 22:30, and
	(e) Tuesday 22:30 - Wednesday 19:15.
	}
	\label{EMS}
\end{figure}

\section{Conclusion}\label{conclusion}
We have presented a method to perform multiple change point analysis of an independent sequence of multivariate observations. 
We are able to consistently detect \textit{any} type of distributional change, and 
do not make any assumptions beyond the existence of the $\alpha{th}$ absolute moment, for some $\alpha\in(0,2)$. 
The proposed methods are able to estimate both the number of change points 
and their locations, thus eliminating the need for prior knowledge or supplementary analysis, unlike the methods presented in 
\cite{Hawkins:2001},  \cite{Lavielle:2006}, or \cite{Fong:2011}. 
Furthermore, this advantage does not come at the expense of additional computational complexity; similar to the previously mentioned methods, 
the proposed approach is $\mathcal O(kT^2)$.

Both divisive and agglomerative versions of this method have been presented. The divisive version hierarchically tests the 
statistical significance of each hierarchically estimated change point, while the agglomerative version proceeds by optimizing a goodness-of-fit statistic. 
Because we have established consistency for the divisive procedure we prefer it in practice, 
even though its computation is dependent on the number of change points that are estimated.


\baselineskip=12pt

 \section*{Acknowledgments}
 We would like to thank the three referees and Associate Editor for their careful review of the manuscript and helpful comments.
The authors sincerely thank Toronto EMS for sharing their data. 
The authors are also grateful to Louis C.\ Segalini for his research assistance in preparing the financial data analysis. 
This work was partially supported by National Science Foundation Grant Number CMMI-0926814.

\bibliographystyle{JASA}
\bibliography{mcpr8}		

\end{document}